\documentclass[11pt]{article}
\usepackage{amssymb,amsmath,amsthm,sectsty,txfonts}
\usepackage[letterpaper,hmargin=1.0in,vmargin=1.0in]{geometry}

\usepackage[small,compact]{titlesec}
\usepackage{xspace}
\usepackage{enumitem}
\usepackage{times}

\usepackage[disable,colorinlistoftodos]{todonotes} 

\usepackage{graphicx}
\usepackage[countmax]{subfloat}
\usepackage[unicode,pdfstartview=FitH]{hyperref}

\usepackage[noend]{algorithmic}
\usepackage{algorithm}

\floatname{algorithm}{Protocol}

\newsubfloat{algorithm}

\usepackage[refpage]{nomencl}
\makenomenclature

\newcommand{\VersionString}{}

\newcommand{\Version}[1]{\renewcommand{\VersionString}{Version #1}}

\Version{$ $Id$ $}

\floatname{algorithm}{Protocol}

\newtheorem{theorem}{Theorem}[section]
\numberwithin{equation}{section}
\numberwithin{figure}{section}

\newtheorem{lemma}[theorem]{Lemma}
\newtheorem{proposition}[theorem]{Proposition}

\newtheorem{claim}[theorem]{Claim}

\theoremstyle{definition}
\newtheorem{definition}[theorem]{Definition}
\newtheorem{assumption}[theorem]{Assumption}
\newtheorem{mechanism}[theorem]{Mechanism}

\newcommand{\ignore}[1]{}

\newcommand{\tc}[1]{\todo[color=green!40,size=\small]{Tal: #1}\xspace}

\newcommand{\ic}[1]{\todo[color=blue!40,size=\small]{Ian: #1}\xspace}
\newcommand{\ici}[1]{\todo[inline,color=blue!40,size=\small]{Ian: #1}\xspace}
\newcommand{\svc}[1]{\todo[color=yellow!40,size=\small]{Salil: #1}\xspace}
\newcommand{\svci}[1]{\todo[inline,color=yellow!40,size=\small]{Salil: #1}\xspace}
\newcommand{\yc}[1]{\todo[color=red!30,size=\small]{Yiling: #1}\xspace}
\newcommand{\yci}[1]{\todo[inline,color=red!30,size=\small]{Yiling: #1}\xspace}

\newcommand{\eps}{\ensuremath{\epsilon}\xspace}

\newcommand{\abs}[1]{\left |{#1} \right |}
\newcommand{\pr}[2][]{\Pr_{#1}\left [ {#2} \right ]}
\newcommand{\defeq}{\ensuremath{\stackrel{def}{=}}}

\newcommand{\remove}[1]{}

\newcommand{\ar}{\rightarrow}

\newcommand{\R}{\mathbb{R}}
\newcommand{\Z}{\mathbb{Z}}
\newcommand{\N}{\mathbb{N}}
\newcommand{\Rand}{\mathcal{R}}
\newcommand{\Exp}{\mathop{\mathrm{E}}}
\newcommand{\cU}{\mathcal{U}}
\newcommand{\KL}{\mathit{KL}}
\newcommand{\Maj}{\mathrm{Maj}}
\newcommand{\Med}{\mathrm{Med}}
\newcommand{\argmax}{\mathop{\mathrm{argmax}}}

\newenvironment{squishenum}{
   \begin{enumerate}
    \setlength{\parskip}{0pt}
    \setlength{\itemsep}{0pt}      \setlength{\parsep}{0pt}
      \setlength{\topsep}{0pt}       \setlength{\partopsep}{0pt}
      \setlength{\leftmargin}{1.5em} \setlength{\labelwidth}{1em}
      \setlength{\labelsep}{0.5em}} {\end{enumerate}}

\newcommand{\squishlist}{
   \begin{list}{$\bullet$}
    { \setlength{\itemsep}{0pt}      \setlength{\parsep}{3pt}
      \setlength{\topsep}{3pt}       \setlength{\partopsep}{0pt}
      \setlength{\leftmargin}{1.5em} \setlength{\labelwidth}{1em}
      \setlength{\labelsep}{0.5em} } }

\newcommand{\squishlisttwo}{
   \begin{list}{$\bullet$}
    { \setlength{\itemsep}{0pt}    \setlength{\parsep}{0pt}
      \setlength{\topsep}{0pt}     \setlength{\partopsep}{0pt}
      \setlength{\leftmargin}{2em} \setlength{\labelwidth}{1.5em}
      \setlength{\labelsep}{0.5em} } }

\newcommand{\squishend}{
    \end{list}  }

\newcommand{\SD}{\ensuremath{\textrm{SD}}\xspace}
\newcommand{\MD}{\ensuremath{\textrm{MD}}\xspace}
\newcommand{\M}{\ensuremath{\mathcal{M}}\xspace}
\newcommand{\Uo}[1][]{\ensuremath{U^{out}_{#1}}\xspace}
\newcommand{\Up}[1][]{\ensuremath{U^{priv}_{#1}}\xspace}
\newcommand{\Vp}[1][]{\ensuremath{V^{priv}_{#1}}\xspace}

\begin{document}

\markboth{Y. Chen et al.}{Truthful Mechanisms for Agents that Value Privacy}

\begin{titlepage}
\title{Truthful Mechanisms for Agents that Value Privacy\footnote{Work begun when all the authors were
at the Harvard Center for Computation and Society, supported in part by a gift from Google, Inc. and by NSF Grant CCF-0915016.}
}

\author{Yiling Chen\thanks{Center for Research on Computation and Society and School of Engineering and Applied Sciences, Harvard University, 33 Oxford Street, Cambridge, MA.  E-mail: \texttt{yiling@seas.harvard.edu}.}
\and
Stephen Chong\thanks{Center for Research on Computation and Society and School of Engineering and Applied Sciences, Harvard University, 33 Oxford Street, Cambridge, MA.  E-mail: \texttt{chong@seas.harvard.edu}. Supported by
NSF Grant No. 1054172.}
\and
Ian A. Kash\thanks{Microsoft Research Cambridge, 7 J J Thomson Ave, Cambridge CB3 0FB, UK.  E-mail: \texttt{iankash@microsoft.com}.}
\and
Tal Moran\thanks{Efi Arazi School of Computer Science, IDC Herzliya.  Email: \texttt{talm@idc.ac.il}.}
\and
Salil Vadhan\thanks{Center for Research on Computation and Society and School of Engineering and Applied Sciences, Harvard University, 33 Oxford Street, Cambridge, MA.  E-mail: \texttt{salil@seas.harvard.edu}}
}

\maketitle

\svci{others need to add affiliations/grants}
\begin{abstract}
Recent work has constructed economic mechanisms that are both truthful and differentially private. In these mechanisms, privacy is treated separately from the truthfulness; it is not incorporated in players' utility functions (and doing so 
has been shown to lead to non-truthfulness in some cases). 
In this work, we propose a new, general way of modelling privacy in players' utility functions.  Specifically, we only assume that if an outcome $o$ has the property that any report of player $i$ would have led to $o$ with approximately the same probability, then
$o$ has small privacy cost to player $i$. We give three mechanisms that are truthful with respect to our modelling of privacy:
for an election between two candidates, for a discrete version of the facility location problem, and for a general social choice problem with
discrete utilities (via a VCG-like mechanism).\svc{changed description of VCG result to mention discreteness}  As the number $n$ of players increases, the social welfare achieved by our mechanisms approaches optimal (as a fraction
of $n$).

\end{abstract}

%
%
%
%
%
%
%
%

\vfill

\thispagestyle{empty} 
\end{titlepage}

\Version{$ $Id$ $}
\section{Introduction}

In this paper, we
examine the interaction between mechanism design and differential privacy.\svc{changed first sentence... previous one
``we continue a recent line of work...'' sounded too derivative to be the first sentence of the paper}
In particular, we explicitly model privacy in players' utility functions and design truthful mechanisms with respect to it.
Our work is motivated by considerations in both fields.
\svci{below I assume that the reader already knows what is ``mechanism design'' and what is ``differential privacy,''
at least at an intuitive level.  I think this is relatively safe for a STOC submission.  But we may consider whether to begin each paragraph with a one or two sentence description of each area}

In mechanism design,
it has long been recognized that players may not behave as predicted due to traditional incentives analysis
out of concerns for privacy:\svc{is there anything we can cite for this? ideally some classic econ paper...}
in addition to having preferences about the outcome of a mechanism (e.g., who wins an auction,
or where a hospital is located), they may also be concerned about what others learn about their private information (e.g., how much they value the auctioned good, or whether they have some medical condition that makes them care more about the hospital's location).  The latter concerns are not modelled in most works on mechanism design, and it is natural to try to bring the new models and techniques of differential privacy to bear on them.

Differential privacy~\cite{DwMcNiSm} is a notion developed to capture privacy when performing statistical analyses of databases.
Informally, a randomized algorithm is differentially private if changing a single individual's data does not ``substantially'' change the output distribution of the algorithm.  Thus, differential privacy is not an absolute notion, but rather
a quantitative one that needs to be weighed against other objectives.  Indeed, differentially private algorithms typically offer a tradeoff between the level of privacy offered to individuals in a database  and the accuracy of statistics computed on the database, which we can think of as a ``global'' objective to be optimized.  However, it is also of interest to consider how privacy should be weighed against the objectives of the individuals themselves.  Mechanism design provides a natural setting in which to consider such tradeoffs.  Attempting to model and reason about privacy in the context of mechanism design seems likely to lead to an improved understanding about the meaning and value of privacy.

\subsection{Previous Work}

The first work bringing together differential privacy and mechanism design was by McSherry and Talwar~\cite{McSherryTa07}.  They showed how to use differential privacy as a {\em tool} for mechanism design.
By definition, differentially private algorithms are insensitive to individuals' inputs; a change in a single individual's
input has only a small effect on the output distribution of the algorithm.  Thus, if a mechanism is differentially
private (and players have bounded utility functions), it immediately follows that the mechanism is {\em approximately truthful}.  That is, reporting untruthfully can only provide a small gain in a player's utility.  With this observation, McSherry and Talwar
showed how tools from differential privacy allow construction of approximately truthful mechanisms for many problems, including ones where exact truthfulness is impossible.

However, as pointed out by Nissim, Smorodinsky, and Tennenholz~\cite{NissimSmTe10}, the approximate truthfulness
achieved by McSherry and Talwar~\cite{McSherryTa07} may not be a satisfactory solution concept. While differential
privacy can guarantee that a player will gain arbitrarily little by lying, it also makes the potential gain
from telling the truth equally small.
Thus players may choose to lie in order to protect their privacy.
Even worse,  as shown by an example in Nissim~et al.~\cite{NissimSmTe10}, in some cases
misreporting is a dominant strategy of the game. Thus, it is difficult to predict the outcome and ``global'' objectives such as social welfare of differentially private mechanisms. Motivated by this, Nissim~et al.~\cite{NissimSmTe10}  show how to
modify some of the mechanisms of McSherry and Talwar~\cite{McSherryTa07} to provide exact truthfulness.  In doing so,
they sacrifice differential privacy.

A recent paper by Xiao~\cite{Xiao11} shows how to remedy this deficiency and construct mechanisms that simultaneously achieve exact truthfulness and
differential privacy.  Xiao's paper also points out that even this combination may not be sufficient for getting players that
value privacy to report truthfully.  Indeed, exact truthfulness only means that a player {\em weakly} prefers to tell the truth.
Lying might not reduce the player's utility at all (and differential privacy implies that it can only reduce the player's utility by at most a small amount).  On the other hand, differential privacy does not guarantee ``perfect'' privacy protection, so it is possible that a player's concern for privacy may still outweigh the small or zero benefit from being truthful.

To address this, Xiao~\cite{Xiao11} advocated incorporating privacy directly into the players' utility functions, and
seeking mechanisms that are truthful when taking the combined utilities into account.  He proposed to measure privacy cost
as the the mutual information between a player's type (assumed to come from some prior distribution) and the outcome of the
mechanism.%
\footnote{Subsequent to our work, Xiao has revised his model to use a different, prior-free measure of privacy.}
Using this measure, he showed that his mechanism does not remain truthful when incorporating privacy into the
utility functions, and left as an open problem to construct mechanisms that do.

\subsection{Our Contributions}\svc{changed ``Results'' to ``Contributions'', since modelling is not a ``result''}

In this paper, we propose a new, more general way of modelling privacy in players' utility functions.  Unlike Xiao's mutual
information measure, our model does not require assuming a prior on players' types, and is instead a pointwise model: we simply assume that if an outcome $o$ has the property that any report of player $i$ would have led to $o$ with approximately the same probability, then
$o$ has small privacy cost to player $i$.   One motivation for this assumption is that such an outcome $o$ will induce
only a small change in a Bayesian adversary's beliefs about player $i$ (conditioned on the other players' reports).
(This is inspired by a Bayesian interpretation of differential privacy, due to Dwork and McSherry and described in \cite{KasiviswanathanSm08}.)
While Xiao's mutual information measure is not strictly a special case of our model, we show (in the appendix) that truthfulness with respect to our modelling implies truthfulness with respect to Xiao's.

We give three mechanisms that are truthful with respect to our model of privacy:
one for an election between two candidates, one for a discrete version of the facility location problem, and one for general social choice problems with
discrete utilities (via a VCG-like mechanism).
As the number $n$ of players increases, the social welfare achieved by our mechanisms approaches optimal (as a fraction
of $n$).

      Our mechanisms
are inspired by Xiao's mechanisms, but with some variations and new analyses to obtain truthfulness when taking privacy into account.  For the election and facility location mechanisms, we can establish {\em universal truthfulness}---truthfulness for every choice of the mechanism's random coins.   For our VCG-like mechanism for general social choice problems, we need to work a bit harder to also ensure that the payments requested do not compromise privacy, and this leads us to only achieve truthfulness in expectation.
In a nutshell, our proofs of universal truthfulness consider two cases for every fixing of the player's reports and coin tosses of the mechanism: If a player misreporting does not affect the outcome of the mechanism, then that player is completely indifferent between truth-telling and misreporting, even taking privacy into account.
On the other hand, if the player misreporting does change the outcome of the mechanism, then being truthful provides a noticeable gain in utility (for the mechanisms we consider)
while differential privacy ensures that the privacy cost of the outcome is still small.  Thus, this
analysis allows us to argue that the benefit of truthfulness outweighs privacy cost even when a player
has a tiny probability of affecting the outcome (e.g., in a highly skewed election using a majority vote
with random noise).  Indeed, our key observation is that the expected privacy cost is also tiny in such case.
\svci{mention individual rationality here?  thought it might be too distracting in the intro}

Unlike previous works, we do not treat differential privacy as an end in itself but rather as a means to incentivize truthfulness from agents that value privacy.  Thus, we do not necessarily need to set the differential privacy parameter $\eps$ to be very small (corresponding to very high privacy, but a larger price in social welfare); we only need to set it small enough so that the privacy costs are outweighed by the agents' preferences for outcomes.
Specifically, our analysis shows that as we decrease $\eps$, agents' ability to affect the outcome falls, but their expected privacy cost falls even faster.  Thus, it is natural to conclude (as we do) that there is some value of $\eps$ (which may be large if the agents care much more about the outcome than their privacy) at which the privacy cost is small enough relative to the benefit that agents are willing to report truthfully.
Moreover,
by
taking agents' value for privacy into account in the incentives analysis, we can have greater confidence that the
agents will actually report truthfully and achieve the approximately optimal social welfare our analysis predicts.

\subsection{Other Related Work}
Independently of our work, 
Nissim, Orlandi, and Smorodinsky~\cite{NiOrSm11privacy} have considered a related way of modelling privacy in players' utilities and constructed truthful mechanisms under their model.  They assume that if {\em all} outcomes $o$ have the property that no player's report affects the probability of $o$ much (i.e., the mechanism is differentially private), then the {\em overall}
privacy cost of the mechanism is small for every player.  This is weaker than our assumption, which requires an analogous bound on the privacy cost
for each specific outcome $o$.  Indeed, Nissim et al.~\cite{NiOrSm11privacy} do not consider a per-outcome model of privacy, and thus do not obtain a reduced privacy cost when a player has a very low probability of affecting the outcome (e.g., a highly skewed election).
However, they require assumptions that give the mechanism an ability to reward players for truthfulness (through their concept of ``agents' reactions,'' which can be restricted by the mechanism).  For example, in the case of an election or poll between two choices (also considered in their paper), they require that a player {\em directly} benefits from reporting their true choice (e.g., in a poll to determine which of two magazines is more popular, a player will receive a copy of whichever magazine she votes for, providing her with positive utility at no privacy cost), whereas we consider a more standard election where the players only receive utility for their preferred candidate winning (minus any costs due to privacy). In general, we consider the standard mechanism design setting where truthfulness is only rewarded through the public outcome (and possibly payments), and this brings out the central tension that our mechanisms need to reconcile: we can only incentivize truthfulness by giving players an influence on the outcome, but such an influence also leads to privacy costs, which may incentivize lying.

Another recent paper that considers a combination of differential privacy and mechanism design is that of Ghosh and
Roth~\cite{GhoshRo11}.  They consider a setting where each player has some private information and some value for its privacy (measured in a way related to differential privacy).  The goal is to design
a mechanism for a data analyst to compute a statistic of the players' private information as accurately as possible, by
purchasing data from many players and then performing a differentially private computation.
In their model, players may lie about their value for privacy, but they cannot provide false data to the analyst.
So they design mechanisms that get players to truthfully report their value for privacy.
In contrast, we consider settings where players may lie about their data (their private types), but where they have a
direct interest in the outcome of the mechanism, which we use to outweigh their value for privacy (so we do not need to explicitly elicit their value for privacy).

Subsequent to our work, Huang and Kannan~\cite{HuangKa} examined the properties of the exponential mechanism~\cite{McSherryTa07}, which can be thought of as noisy version of VCG that is slightly different from the one we study.  They showed that, with appropriate payments, this mechanism is truthful, individually rational, approximately efficient, and differentially private, but their model does not incorporate privacy costs into players' utility functions.

We remark that there have also been a number of works that consider secure-computation-like notions of privacy for
mechanism design problems (see~\cite{NPS99privacy,DoHaRa00crypto-game,IzMiLe05rational,PRST08auctions,BrSa08auctions,FJS10privacy} for some examples).\tc{Why is Dolev-Dwork-Naor relevant?}\svc{I thought one of the motivations for non-malleability was to allow auctions without a secure communication channel to the auctioneer.  but feel free to omit if you think it's
too much of a stretch}
In these works, the goal is to ensure that a distributed implementation of a mechanism does not leak much more information than a centralized implementation by a trusted third party
In our
setting, we assume we have a trusted third party to implement the mechanism and are concerned with the information leaked by the outcome itself.

\svci{other works we should discuss/cite?  from the econ literature?}

\svci{add an Organization section?}

\Version{$ $Id$ $}
\section{Background on Mechanism Design}\label{sec:background}

In this section, we introduce the standard framework of mechanism design to lay the ground for modelling privacy in the context of mechanism design in next section. We use a running example of an election between two candidates.
A (deterministic) mechanism is given by the following components: 
\squishlist
\item A number $n$ of players.  These might be the $n$ voters in an election between two candidates $A$ and $B$.
\item A set $\Theta$ of player types.  In the election example, we take $\Theta=\{A,B\}$, where $\theta_i\in \Theta$ indicates which of the two candidates is preferred by voter $i\in [n]$.
\item A set $O$ of outcomes.  In the election example, we take $O=\{A,B\}$, where the outcome indicates which of the two candidates win. (Note that we do not include the tally of the vote as part of the outcome.  This turns out to be significant for privacy.)
\item Players' action spaces $X_i$ for all $i \in [n]$. In general, a player's action space can be different from his type space. However, in this paper we view the types in $\Theta$ to be values that we expect players to know and report. Hence, we require $X_i = \Theta$ for all $i \in [n]$ (i.e., we restrict to direct revelation mechanisms, which is without loss of generality).  In the election example, the action of a player is to vote for $A$ or for $B$.
\item An outcome function $\M: X_1 \times \dots \times X_n \rightarrow O$ that determines an outcome given players' actions. Since we require $X_i = \Theta$, the outcome function becomes $\M: \Theta^n \rightarrow O$. For example, a majority voting mechanism's function maps votes of players to the candidate who received a majority of votes.
\item Player-specific {\em utility functions}
$U_i :\Theta \times  O \ar \mathbb{R}$ for $i=1,\ldots,n$, giving the utility of player $i$ as a function of his type and the outcome.
\squishend
%

To simplify notation, we use a mechanism's outcome function to represent the mechanism. That is, a mechanism is denoted $\M: \Theta^n \rightarrow O$.
The goal of mechanism design is then to design a {\em mechanism} $\M : \Theta^n \rightarrow O$
that takes players' (reported) types and selects an outcome so as to maximize some global objective function (e.g.
the sum of the players' utilities, known as {\em social welfare}) even when players may falsely report their type in order to increase their personal utility.  The possibility of players' misreporting is typically handled by designing mechanisms that
are {\em incentive-compatible}, i.e., it is in each player's interest to report their type honestly.  A strong formulation of
incentive compatibility is the notion of {\em truthfulness} (a.k.a. dominant-strategy incentive compatibility):
for all
players $i$, all  types $\theta_i\in \Theta$, all alternative reports $\theta_i'\in \Theta$, and all profiles $\theta_{-i}$ of
the other players' reports\footnote{We adopt the standard game-theory convention that $\theta_{-i}$ refers to all components of the vector $\theta$ except the one corresponding to player $i$, and that $(\theta_i,\theta_{-i})$ denotes the vector obtained by putting $\theta_i$ in the $i$'th component and using $\theta_{-i}$ for the rest.}, we have:
\begin{equation} \label{eqn:truthfulness}
U_i(\theta_i,\M(\theta_i,\theta_{-i})) \geq U_i(\theta_i,\M(\theta'_i,\theta_{-i})).
\end{equation}
If Inequality (\ref{eqn:truthfulness}) holds for player $i$ (but not necessarily all players), we say that
the mechanism is {\em truthful for player $i$}.
Note that we are using $\theta_{-i}$ here as both the type and the report of other players.  Since truthfulness must hold for all possible reports of other players, it is without loss of generality to assume that other players report their true type.
This is in contrast to the notion of a Nash equilibrium which refers to the incentives of player $i$ under the assumption that other players are using equilibrium strategies. 

In the election example, it is easy to see that standard majority voting is a truthful mechanism.  Changing one's vote
to a less-preferred candidate can never increase one's utility (it either does not affect the outcome, or does so in a way
that results in lower utility). 

In this paper, we will allow randomized mechanisms, which we define as
$\M : \Theta^n\times \Rand \rightarrow O$, where $\Rand$ is the probability space from which the mechanism makes its
random choices (e.g., all possible sequences of coin tosses used by the mechanism). 
We write $\M(\theta)$ to denote the random variable obtained by sampling $r$ from $\Rand$ and evaluating $\M(\theta;r)$.
This (non-standard) definition of a randomized mechanism is equivalent to the standard one (where the mechanism is a function from reported types to a distribution over outcomes) and makes our analysis clearer.

For randomized mechanisms, one natural generalization of
truthfulness is {\em truthfulness in expectation}:
for all
players $i$, all types $\theta_i$, all utility functions $U_i$, all reports $\theta_i'$, and all profiles $\theta_{-i}$ of
the other players' reports, we have:
\[\Exp[U_i(\theta_i,\M(\theta_i,\theta_{-i}))] \geq \Exp[U_i(\theta_i,\M(\theta'_i,\theta_{-i}))],
\]
where the expectation is taken over the random choices of the mechanism.

A stronger notion is that of {\em universal truthfulness}:
for all
players $i$, all types $\theta_i$ and utility functions $U_i$, all alternative reports $\theta_i'$, and all profiles $\theta_{-i}$ of
the other players' reports, and all $r\in \Rand$, we have:
\[U_i(\theta_i,\M(\theta_i,\theta_{-i};r)) \geq U_i(\theta_i, \M(\theta'_i,\theta_{-i};r)).
\]
Thus $\M$ being universally truthful is equivalent to saying that for every $r\in \Rand$, $\M(\cdot;r)$ is a deterministic
truthful mechanism.

\Version{$ $Id$ $}
\section{Modelling Privacy in Mechanism Design}\label{sec:modelprivacy}

The standard framework of mechanism design does not consider a player's value of privacy. In this section, we incorporate privacy into mechanism design and adapt the definitions of truthfulness accordingly.
%
%
We continue considering the basic mechanism-design setting from Section~\ref{sec:background}.
However, players now care not only about the outcome of the mechanism, but also what that outcome reveals about their private types.  Thus, a player's utility becomes
\begin{equation}
\label{eqn:quasilinear}
U_i = \Uo[i]+\Up[i],
\end{equation}
where $\Uo[i] : \Theta \times O\rightarrow \R$ is player $i$'s utility for the outcome
and $\Up[i]$ is player $i$'s utility associated with privacy or information leakage.  
Before discussing the form of $\Up$ (i.e., what are its inputs), we note that in
Equation~(\ref{eqn:quasilinear}), there is already an implicit assumption that privacy can be measured in units that
can be linearly traded with other forms of utility.  A more general formulation would allow $U_i$ to be an arbitrary monotone
function of $\Uo[i]$ and $\Up[i]$, but we make the standard quasi-linearity assumption for simplicity.

Now, we turn to functional form of $\Up[i]$.  First, we note that $\Up[i]$ should not just be a function of player $i$'s type and the outcome.  What matters is the {\em functional relationship} between player $i$'s reported type and the outcome.  For example, a voting mechanism that ignores player $i$'s vote should have zero privacy cost to player $i$, but one that uses player $i$'s vote to entirely determine the outcome may have a large privacy cost.  So we will allow $\Up[i]$ to depend on the mechanism itself,
as well as the reports of other players, since these are what determine the functional relationship between player $i$'s report and the outcome:
\begin{equation}
\Up[i] : \Theta \times  O \times \{\M : \Theta^n\times \Rand \rightarrow O\} \times \Theta^{n-1} \rightarrow \R.
\end{equation}

Thus, when the reports of the $n$ players are $\theta'\in \Theta^n$ and the outcome is $o$, the
utility of player $i$ is $$U_i(\theta_i,o,\M,\theta'_{-i})=\Uo[i](\theta_i,o)+\Up[i](\theta_i,o,\M,\theta'_{-i}).$$
 \yc{I changed the report to $\theta'$ to differentiate the reported type and true type. Someone may want to check to make sure that I got it correct.}
 \svc{It looks correct to me, but if we use $\theta'_{-i}$ here, then shouldn't we make analogous changes many places below, eg in Assumption 3.1, Definition 3.2, Definition 3.4?  Those might get hard to read with so many primed variables (then we'll need to use double and triple primes for alternative reports of player $i$)}
\ic{I left this alone for now, but think we can revert it if we want to / have time}
In particular, $U_i$ has the same inputs as $\Up$ above, including $\M$.
Unlike standard mechanism design, we are not given fixed utility functions and then need to design a mechanism with respect to those utility functions.  Our choice of mechanism
affects the utility functions too!

Note that we do not assume that $\Up[i]$ is always negative (in contrast to Xiao~\cite{Xiao11}).  In some cases,
players may prefer for information about them to be kept secret and in other cases they may prefer for it to be leaked (e.g., in case it is flattering).  Thus, $\Up[i]$ may be better thought of as ``informational utility'' rather than a ``privacy cost''.


It is significant that we do not allow the $\Up[i]$ to depend on the {\em report} or, more generally, the {\em strategy} of player $i$.  This is again in contrast to Xiao's modelling of privacy~\cite{Xiao11}.  We will discuss the motivation for our choice in Section~\ref{sec:discussion}, and also show that despite this difference, truthfulness with respect to our modelling implies truthfulness with respect to Xiao's modelling (Appendix~\ref{sec:xiao}).

Clearly no mechanism design would be possible if we make no further assumptions about the $\Up[i]$'s and allow them to be
arbitrary, unknown functions (as their behavior could completely cancel the $\Uo[i]$'s).  Thus, we will make the natural assumption that $\Up[i]$ is small if player $i$'s report has little influence on the outcome $o$.
More precisely:
\begin{assumption}[privacy-value assumption]\svc{added subscript of $i$ to $\theta'_i$ and $\theta''_i$ here and in Def 3.3}
\label{ass:generic-privacy}
\begin{equation*}
\forall \theta \in\Theta^n,o\in O,\M : \abs{\Up[i](\theta_i,o,\M,\theta_{-i})} \le
  F_i\left(\max_{\theta'_i,\theta''_i\in\Theta}
  \frac{\pr{\M(\theta'_i,\theta_{-i})=o}}{{\pr{\M(\theta''_i,\theta_{-i})=o}}}\right)\ ,
\end{equation*}
where $F_i : [1,\infty)\rightarrow [0,\infty]$ is a {\em privacy-bound}\svc{better name?} function with the property that $F_i(x)\rightarrow 0$ as $x\rightarrow 1$, and the probabilities are taken over the random choices of $\M$.\svc{we have a choice of whether
to put $\theta_{-i}$ in the $\forall$ or the $\max$.  I don't have a strong preference, but I found it easier to motivate
as above.}
\end{assumption}

Note that if the mechanism ignores player $i$'s report, then the right-hand side of (\ref{ass:generic-privacy}) is $F_i(1)$, which naturally corresponds to a privacy cost of 0.  Thus, we are assuming that the privacy costs satisfy a continuity condition as the mechanism's dependence on player $i$'s
report decreases.   The privacy-bound function $F_i$ could be the same for all players, but we allow it
to depend on the player for generality.

Assumption~(\ref{ass:generic-privacy}) is inspired by the notion of {\em differential privacy}, which is due to Dinur and Nissim~\cite{DinurNi}, Dwork and Nissim~\cite{DworkNi}, Blum~et al.~\cite{BlumDwMcNi}, and Dwork~et~al.~\cite{DwMcNiSm}. We restate it in our
notation:
\begin{definition}\label{def:dp}
A mechanism $\M : \Theta^n \times \Rand \rightarrow O$ is {\em $\eps$-differentially private} iff
$$\forall \theta_{-i}\in\Theta^{n-1},o\in O
\qquad \max_{\theta_i',\theta_i''\in\Theta}\frac{\pr{\M(\theta_i',\theta_{-i})=o}}{{\pr{\M(\theta_i'',\theta_{-i})=o}}}
\leq e^{\eps}.$$
\end{definition}

By inspection of Assumption~(\ref{ass:generic-privacy}) and the definition of differential privacy, we have the following result.
\begin{proposition} \label{prop:dp}
If $\M$ is $\eps$-differentially private, then for all players $i$ whose utility functions satisfy Assumption~(\ref{ass:generic-privacy}), all $\theta_{-i}\in \Theta^{n-1}$, and $o\in O$, we have
$\abs{\Up[i](\theta_i,o,\M,\theta_{-i})} \le F_i(e^\eps).$
\end{proposition}
In particular, as we take $\eps\rightarrow 0$, the privacy cost of any given outcome tends to 0.

Like differential privacy, Assumption~(\ref{ass:generic-privacy}) makes sense only for randomized mechanisms,
and only measures the loss in privacy contributed by
Player $i$'s report when fixing the reports of the other players.  In some cases, it may be that the other players' reports already reveal
a lot of information about player $i$.  See Section~\ref{sec:discussion} for further discussion, interpretation, and critiques of our modelling.
%
With this model, the definitions of truthfulness with privacy are direct analogues of the basic definitions given earlier.

\begin{definition}[truthfulness with privacy]
Consider a mechanism design problem
with $n$ players, type space $\Theta$, and outcome space $O$.  For a player $i$ with utility function $U_i = \Uo[i]+\Up[i]$, we say that a randomized mechanism $\M : \Theta^n \times \Rand \rightarrow O$ is
{\em truthful in expectation for player $i$}
if for all types $\theta_i\in \Theta_i$, all alternative reports $\theta_i'\in \Theta$ for player $i$, and all possible profiles $\theta_{-i}$ of
the other players' reports, we have:
\[\Exp[U_i(\theta_i,\M(\theta_i,\theta_{-i}),\M,\theta_{-i})] \geq \Exp[U_i(\theta_i,\M(\theta'_i,\theta_{-i}),\M,\theta_{-i})].
\]
We say that $\M$ is {\em universally truthful for player $i$}
if the inequality further holds for all values of $r \in \Rand$:
\[U_i(\theta_i,\M(\theta_i,\theta_{-i};r),\M,\theta_{-i}) \geq U_i(\theta_i,\M(\theta'_i,\theta_{-i}; r),\M,\theta_{-i}).
\]
\end{definition}
Note that, unlike in standard settings, $\M$ being universally truthful does {\em not} mean
that the deterministic mechanisms $\M(\cdot;r)$ are truthful.  Indeed, even when we fix $r$, the privacy utility $\Up[i](\theta,o,\M,\theta_{-i})$ still depends on the original randomized function $\M$, and the privacy properties of $\M$ would be lost
if we publicly revealed $r$.  What universal truthfulness means is that player $i$ would still want to report truthfully even if she knew $r$ but it were kept secret from the rest of the world.

\section{Private Two-Candidate Elections}

Using Proposition~\ref{prop:dp}, we will sometimes be able to obtain truthful mechanisms taking privacy into account by applying tools from differential
privacy to mechanisms that are already truthful when ignoring privacy.
In this section, we give an illustration of this approach in our example of a two-candidate election.


\begin{mechanism}\label{mech:election}
Differentially private election mechanism

\noindent Input: profile $\theta\in \{A,B\}^n$ of votes, privacy parameter $\eps>0$.
\begin{squishenum}
\item Choose $r\in \Z$ from a discrete Laplace distribution, namely $\Pr[r=k] \propto \exp(-\eps|k|)$.\svc{the mechanism would
be more symmetric if we chose $r$ to be an odd multiple of $1/2$, so that we don't advantage $A$ in case of a tie... maybe something to do in a future version.}
\item If $\#\{ i : \theta_i = A\} - \#\{ i : \theta_i = B\} \geq r$, output $A$.  Otherwise output $B$.
\end{squishenum}
\end{mechanism}

We show that for sufficiently small $\eps$, this mechanism is truthful for players satisfying
Assumption~\ref{ass:generic-privacy}:
\begin{theorem} \label{thm:voting}
Mechanism~\ref{mech:election} is universally truthful
for player $i$ provided that, for some function $F_i$:
\begin{squishenum}
\item Player $i$'s privacy utility $\Up[i]$ satisfies Assumption~\ref{ass:generic-privacy} with privacy bound function $F_i$, and
\item $\Uo[i](\theta_i,\theta_i)-\Uo[i](\theta_i,\neg \theta_i) \geq 2 F_i(e^\eps),$ \label{cond:UtilityVsPrivacy}
\end{squishenum}
\end{theorem}
Note that Condition~\ref{cond:UtilityVsPrivacy} holds for sufficiently small $\epsilon>0$ (since $F_i(x)\rightarrow 0$ as $x\rightarrow 1$).  The setting of $\eps$ needed to achieve truthfulness depends only on how much the players value their preferred candidate
(measured by the left-hand side of Condition~\ref{cond:UtilityVsPrivacy}) and how much they value privacy (measured by the right-hand side of Condition~\ref{cond:UtilityVsPrivacy}), and is independent of the number of players $n$.

\begin{proof}
Fix the actual type $\theta_i\in \{A,B\}$ of player $i$, a profile $\theta_{-i}$ of reports of the
other players, and a choice $r$ for $\M$'s randomness.  The only alternate report for player $i$ we need to consider is
$\theta_i' = \neg \theta_i$. Let $o=\M(\theta_i,\theta_{-i};r)$ and $o' = \M(\neg \theta_i,\theta_{-i};r)$.
We need to show that
$U_i(\theta_i,o,\M,\theta_{-i}) \geq U_i(\theta_i,o',\M,\theta_{-i})$, or equivalently
\begin{equation} \label{eqn:deviation}
\Uo[i](\theta_i,o)-\Uo[i](\theta_i,o') \geq
\Up[i](\theta_i,o',\M,\theta_{-i})-\Up[i](\theta_i,o,\M,\theta_{-i}).
\end{equation}
We consider two cases:
\begin{description}
\setlength{\itemsep}{0pt}
\item[Case 1: $o=o'$]  In this case, Inequality~(\ref{eqn:deviation}) holds because both the left-hand and right-hand sides
are zero.

\item[Case 2: $o\neq o'$]  This implies that $o=\theta_i$ and $o'=\neg\theta_i$.  (If player $i$'s report has any effect
on the
outcome of the differentially private voting mechanism, then it must be that the outcome equals player $i$'s report.)
Thus the left-hand side of Inequality~(\ref{eqn:deviation}) equals $\Uo[i](\theta_i,\theta_i)-\Uo[i](\theta_i,\neg \theta_i)$. By
Proposition~\ref{prop:dp}, the right-hand side of Inequality (\ref{eqn:deviation}) is at most $2F(e^\eps)$.  Thus,
Inequality (\ref{eqn:deviation}) holds by hypothesis. \qedhere
\end{description}
\end{proof}


Of course, truthfulness is not the only property of interest.  After all, a mechanism that is simply a constant function is (weakly) truthful.  Another property we would like is economic {\em efficiency}.
Typically, this is defined as maximizing social welfare, the sum of players' utilities.  Here we consider the sum of {\em outcome} utilities
for simplicity.
As is standard, we normalize players' utilities so that all players are counted equally in measuring the social welfare.  In our voting example, we
wish to maximize the number of voters' whose preferred candidates win, which is equivalent to normalizing the
left-hand side of Condition~\ref{cond:UtilityVsPrivacy} in Theorem~\ref{thm:voting} to 1.  Standard, deterministic majority voting clearly maximizes this
measure of social welfare.   Our mechanism achieves approximate efficiency:

\svci{added expected social welfare to prop and tightened proof.  would be good for someone to check}
\begin{proposition} \label{prop:voting-efficiency}
For every profile $\theta\in \Theta^n$ of reports, if we select $o\leftarrow \M(\theta)$ using Mechanism~\ref{mech:election}, then:
\begin{squishenum}
\item $\Pr\left[\#\{ i : \theta_i=o\} \leq \max_{o'\in \{A,B\}} \#\{ i : \theta_i = o'\} - \Delta\right] < e^{-\eps\Delta}$.

\item $\Exp\left[\#\{ i : \theta_i=o\}\right] > \max_{o'\in \{A,B\}} \#\{ i : \theta_i = o'\} - 1/\eps$.
\end{squishenum}
\end{proposition}

\begin{proof}
The maximum number of voters will be satisfied by taking the majority candidate $o^* = \Maj(\theta)$, where we break ties in
favor of $A$.  Let $\Delta' = \#\{ i : \theta_i=o^*\}-\#\{ i : \theta_i=\neg o^*\}$.
If $o^*=A$, then
$\neg o^* = B$ is selected iff the noise $r$ is larger than $\Delta'$.  If $o^*=B$, then
$\neg o^* = A$ is selected iff the noise $r$ is smaller than or equal to $-\Delta'$.  Since $r$ is chosen so that $\Pr[r=k] \propto e^{-\eps|k|}$, the probability of selecting $\neg o^*$ in either case
is bounded as:
$$\Pr[\M(\theta)=\neg o^*] \leq \frac{\sum_{k\geq \Delta'} e^{-\eps k}}{\sum_{k\in \Z} e^{-\eps |k|}} = \frac{e^{-\eps\Delta'}}{1+e^{-\eps}}
\leq e^{-\eps\Delta'}.$$
Now the high probability bound follows by considering the case that $\Delta'\geq \Delta$ (otherwise the event occurs
with probability 0).  The expectation bound is computed as follows:
\begin{eqnarray*}
&\Exp\left[\max_{o'\in \{A,B\}} \#\{ i : \theta_i = o'\}-\#\{ i : \theta_i=\M(\theta)\}\right]\\
&= \Pr[\M(\theta)=\neg o^*]\cdot \Delta'
\leq  \Delta'\cdot \frac{e^{-\eps\Delta'}}{1+e^{-\eps}}
\leq \frac{1}{e\eps}\cdot\frac{1}{1+e^{-\eps}}
< \frac{1}{\eps},
\end{eqnarray*}
where the second-to-last inequality follows from the fact that
$xe^{-\eps x}$ is minimized at $x=1/\eps$.
\end{proof}

\svci{following two pars are new}
Thus, the number of voters whose preferred candidate wins is within $O(1/\eps)$ of optimal, in expectation
and with high probability.
This deviation is independent of $n$, the number of players.  Thus if we take $\eps$ to be a
constant (as suffices for
truthfulness) and let $n\rightarrow \infty$, the economic efficiency approaches optimal, when we consider both as fractions of $n$.  This also holds for
vanishing $\eps=\eps(n)$, provided $\eps=\omega(1/n)$. 
Despite the notation, $\epsilon$ need not be tiny, which allows a tradeoff between efficiency and privacy.

This analysis considers the social welfare as a sum of outcome utilities (again normalizing so that everyone values their preferred candidate by one unit of utility more than the other candidate).  We can consider the effect of privacy utilities on the social welfare too.  By Proposition~\ref{prop:dp}, the privacy utilities affect the social welfare by at most $\sum_i F_i(e^\eps)$, assuming player $i$ satisfies Assumption~\ref{ass:generic-privacy} with privacy bound function $F_i$. If all players satisfy Assumption~\ref{ass:generic-privacy} with the same privacy bound function $F_i=F$, then the effect on social welfare
is at most $n\cdot F(e^\eps)$.  By taking $\eps\rightarrow 0$ (e.g., $\eps=1/\sqrt{n}$), the privacy utilities contribute a vanishing fraction of $n$.

Another desirable property is {\em individual rationality}: players given the additional option of not participating should still prefer to participate and report truthfully.  This property follows from the same argument we used to establish universal truthfulness.  By dropping out, the only change in outcome
that player $i$ can create is to make her less preferred candidate win.  Thus, the same argument as in Theorem~\ref{thm:voting} shows that player $i$ prefers truthful participation to dropping out.

\begin{proposition}
Under the same assumptions as Theorem~\ref{thm:voting}, Mechanism~\ref{mech:election} is individually rational for player $i$.
\end{proposition}

\section{Tools for Proving Truthfulness with Privacy}

The analysis of truthfulness in Theorem~\ref{thm:voting} is quite general.  It holds for any differentially
private mechanism with the property that if  a player can
actually change the outcome of the mechanism by reporting untruthfully, then it will have a noticeable negative impact on the player's outcome utility.
We abstract this property for use in analyzing our other mechanisms.
\begin{lemma} \label{lem:universal}
Consider a mechanism design problem
with $n$ players, type space $\Theta$, and outcome space $O$. Let player $i$ have
a utility function $U_i = \Uo[i]+\Up[i]$ satisfying Assumption~\ref{ass:generic-privacy} with privacy bound function $F_i$.
Suppose that randomized mechanism $\M : \Theta^n\rightarrow O$ has the following properties:
\begin{squishenum}
\item $\M$ is $\eps$-differentially private, and
\item For all possible types $\theta_i$, all profiles $\theta_{-i}$
of the other players' reports, all random choices $r$ of $\M$, and all alternative reports $\theta_i'$ for player $i$:
if $\M(\theta_i,\theta_{-i};r)\neq \M(\theta_i',\theta_{-i};r)$, then
$\Uo(\theta_i,\M(\theta_i,\theta_{-i};r)) - \Uo(\theta_i,\M(\theta_i',\theta_{-i};r)) \geq 2F_i(e^\eps)$,
\end{squishenum}
Then $\M$ is universally truthful for player $i$.
\end{lemma}

Lemma~\ref{lem:universal} implicitly requires that players not be indifferent between outcomes 
(unless there is no way for the player to change one outcome to the other).  A condition like this is necessary
because otherwise players may be able to find reports that have no effect on their outcome utility but improve their privacy utility.  However, we show in Section~\ref{sec:vcg} that in some settings with payments, truthfulness can be achieved even with indifference between outcomes since payments can break ties.

It is also illustrative and useful to consider what happens when we take the expectation over the mechanism's coin tosses.
We can upper-bound the privacy utility as follows:
\begin{lemma} \label{lem:SD}
Consider a mechanism design problem
with $n$ players, type space $\Theta$, and outcome space $O$. Let player $i$ have type $\theta_i\in \Theta_i$ and a
utility function $U_i = \Uo[i]+\Up[i]$ satisfying Assumption~\ref{ass:generic-privacy}.
Suppose that randomized mechanism $\M : \Theta^n \times \Rand \rightarrow O$ is $\eps$-differentially private.
Then
for all possible profiles $\theta_{-i}$
of the other players' reports, all random choices $r$ of $\M$, and all alternative reports $\theta_i'$ for player $i$,
we have
\begin{equation*}
\abs{\Exp[\Up[i](\theta_i, \M(\theta_i,\theta_{-i}),\M,{\theta_{-i}})] -
\Exp[\Up[i](\theta_i, \M(\theta_i',\theta_{-i}),\M,{\theta_{-i}})]}
\leq 2F_i(e^\eps)\cdot\SD(\M(\theta_i,\theta_{-i}),\M(\theta_i',\theta_{-i})),
\end{equation*}
where $\SD$ denotes statistical difference.\footnote{The {\em statistical difference} (aka {\em total variation distance})
between
two discrete random variables $X$ and $Y$ taking values in a universe $\cU$ is defined to be $\SD(X,Y) = \max_{S\subseteq \cU} |\Pr[X\in S]-\Pr[Y\in S]|$.}
\end{lemma}
\begin{proof}
For every two discrete random variables $X$ and $Y$ taking values in a universe $\cU$, and every function
$f : \cU\rightarrow [-1,1]$, it holds that $|\Exp[f(X)]-\Exp[f(Y)]| \leq 2\SD(X,Y)$.  (The $f$ that maximizes the left-hand side sets $f(x)=1$ when $\Pr[X=x]>\Pr[Y=x]$ and sets $f(x)=-1$ otherwise.)
Take $\cU=O$, $X=\M(\theta_i,\theta_{-i})$, $Y=\M(\theta_i',\theta_{-i})$, and $f(o) = \Up[i](\theta_i,o,\M,{\theta_{-i}})/F_i(e^\eps)$. By Proposition~\ref{prop:dp}, we have $f(o)\in [-1,1]$.
\end{proof}

By this lemma, to establish truthfulness in expectation, it suffices to
show that the expected gain in outcome utility from reporting $\theta_i$ instead of
$\theta_i'$ grows proportionally with
$\SD(\M(\theta_i,\theta_{-i}),\M(\theta_i',\theta_{-i}))$.
(Specifically, it should be at least the statistical difference
times $2F_i(e^\eps)$.)
In Lemma~\ref{lem:universal}, the gain in outcome utility is related to
the statistical difference by coupling the random variables $\M(\theta_i,\theta_{-i})$
and $\M(\theta_i',\theta_{-i})$ according to the random choices $r$ of $\M$.  Indeed,
$\Pr_r[\M(\theta_i,\theta_{-i};r)\neq \M(\theta'_i,\theta_{-i};r)] \geq
\SD(\M(\theta_i,\theta_{-i}),\M(\theta_i',\theta_{-i})).$
Thus, if the outcome-utility gain from truthfulness is larger than $2F_i(e^\eps)$ whenever $\M(\theta_i,\theta_{-i};r)\neq \M(\theta'_i,\theta_{-i};r)$, then we have truthfulness in expectation (indeed, even universal truthfulness).

We note that if $\M$ is differentially private, then
$\SD(\M(\theta_i,\theta_{-i}),\M(\theta_i',\theta_{-i})) \leq e^\eps-1 = O(\eps),$
for small $\eps$.  By Lemma~\ref{lem:SD}, the expected difference in privacy utility between any two reports is
at most $O(F_i(e^\eps)\cdot \eps)$.  Thus, $\eps$-differential privacy helps us twice, once in bounding the pointwise
privacy cost (as $2F_i(e^\eps)$), and second in bounding the statistical difference between outcomes.  On the other hand,
for mechanisms satisfying the conditions of Lemma~\ref{lem:universal}, the differential privacy only affects the expected outcome utility by a factor related to the statistical difference.  This is why, by taking $\eps$ sufficiently small, we
can ensure that the outcome utility of truthfulness dominates the privacy cost.

Lemma~\ref{lem:SD} is related to, indeed inspired by, existing lemmas used to analyze the composition of differentially
private mechanisms.  These lemmas state that while differential privacy guarantees a worst case bound of $\eps$ on the ``privacy loss'' of all possible outputs, this actually implies an {\em expected} privacy loss of $O(\eps^2)$.  Such bounds
correspond to the special case of Lemma~\ref{lem:SD} when $F_i = \ln$ and we replace the statistical difference with
the upper bound $e^\eps-1$.
These $O(\eps^2)$ bounds on expected privacy loss were
proven first in the case of specific mechanisms by Dinur and Nissim~\cite{DinurNi} and Dwork and Nissim~\cite{DworkNi}, and then in the case of arbitrary differentially private mechanisms by Dwork~et al.~\cite{DworkRoVa10}.  In our case, the $O(\eps^2)$ bound does not suffice, and we need the stronger
bound expressed in terms of the statistical difference. Consider the differentially
private election when the vote is highly skewed (e.g., 2/3 vs. 1/3).  Then a player has only an exponentially small probability (over the random choice $r$ of the mechanism) of affecting the outcome, and so the expected outcome utility for voting truthfully is exponentially small.  On the other hand, by Lemma~\ref{lem:SD},
the expected privacy loss is also exponentially small, so we can still have
truthfulness.

\Version{$ $Id$ $}
\section{Discrete Facility Location}\label{sec:facility}

In this section, we apply our framework to discrete facility location.  Let $\Theta = \{\ell_1<\ell_2<\cdots<\ell_q\} \subset [0,1]$ be a finite set of types indicating player's preferred locations for a facility on the unit interval and $O = [0,1]$.
Players prefer to have the facility located as close to them as possible: $\Uo[i](\theta_i,o) = - |\theta_i - o|$.   For example, the mechanism may be selecting a location
for a bus stop along a major highway, and the locations $\ell_1,\ldots,\ell_q$ might correspond to cities along the highway where potential bus riders
live.\svc{added example to justify discreteness... feel free to edit to make it more natural/realistic}

Note that the voting game we previously considered can be represented as the special case where $\Theta = \{ 0,1 \}$.  This problem has a well-known truthful and economically
efficient mechanism: select the location of the median report.  Xiao~\cite{Xiao11} gave a private and truthful mechanism for this problem based on taking the median of a perturbed histogram.  His analysis only proved that the mechanism satisfies ``approximate differential privacy'' (often called $(\eps,\delta)$ differential privacy).
To use Proposition~\ref{prop:dp}, we need the mechanism to satisfy pure $\eps$ differential privacy (as in Definition~\ref{def:dp}).
Here we do that for a variant of Xiao's mechanism.\\ 

\begin{mechanism} \label{mech:dfac}
Differentially private discrete facility location mechanism

\noindent Input: profile $\theta\in \Theta^n$ of types, privacy parameter $\eps>0$.
\begin{squishenum}
\item Construct the histogram $h = (h_1,\ldots,h_q)$ of reported type frequencies where $h_j$ is the number of reports $\theta_i$ of type $\ell_j$ and $q = |\Theta|$.
\item Choose a random (nonnegative, integer) noise vector $r = (r_1,\ldots,r_q)\in \N^q$ where the components $r_j$ are chosen
independently such that $\pr{r_j = k}$ is proportional to $\exp(-\eps k/2)$.
\item Output the type  corresponding to median of the perturbed histogram $h + r$.  That is, we output $\ell_{\Med(h+r)}$, where for $z\in \N^q$ we define
$\Med(z)$ to be the minimum $k\in [q]$ such that $\sum_{j=1}^k z_j \geq \sum_{j = k+1}^q z_j$).
\end{squishenum}
\end{mechanism}

Xiao's mechanism instead chooses the noise components $r_j$ according to a truncated and shifted Laplace distribution.  Specifically,
$\Pr[r_j=k]$ is proportional to $\exp((\eps/2)\cdot |k-t|)$ for $k=0,\ldots,2t$ and $\Pr[r_j=k]=0$ for $k>2t$,
where $t=\Theta(\log(1/\delta)/\eps)$.  This ensures that the noisy histogram $h+r$
is $(\eps,q\delta)$ differentially private, and hence the outcome $\ell_{\Med(h+r)}$ is as well.  Our proof directly analyzes the median,
without passing through the histogram.  This enables us to achieve pure $\eps$ differential privacy and use a simpler noise distribution.
On the other hand, Xiao's analysis is more general, in that it applies to any mechanism that computes its
result based on a noisy histogram.

\begin{lemma} \label{lem:dfacisdp}
Mechanism~\ref{mech:dfac} is $\eps$-differentially private.
\end{lemma}

\begin{proof}
Differential privacy requires that on any pair of histograms $h,h'$ reachable by one player reporting different types,
the probability of any particular outcome $o=\ell_j$ being selected differs by at most an $e^\eps$ multiplicative factor.
Since reporting a different type results in two changes to the histogram (adding to one type and subtracting from another),
we show that on each such change the probability differs by at most an $e^{\eps/2}$ factor.

Consider two histograms $h$ and $h'$ that differ only by an addition or subtraction of 1 to a single entry.  Let $f_j: \N^q \rightarrow \N^q$ map a vector
$s$ to the vector $(s_j+1,s_{-j})$ (i.e., identical except $s_j$ has been increased by 1).  $f_j$ is an injection and has the property that if $j$ is the
median of $h+s$ then $j$ is also the median of
$h'+f_j(s)$.  Note that under our noise distribution, we have $\Pr[r=s] = e^{\eps/2}\cdot  \Pr[r=f_j(s)]$.

Then writing $\M$ as a function of $h$ rather than $\theta$, we have:
\begin{align*}
\pr{\M(h) = \ell_j} &= \sum_{s \text{ s.t. } \Med(h+s)=j} \pr{r=s}\\
&= \sum_{s \text{ s.t. } \Med(h+s)=j} e^{\eps/2}\cdot\pr{r=f_j(s)} \\
&\leq e^{\eps/2}\cdot \sum_{s \text{ s.t. } \Med(h'+f_j(s))=j}\pr{r=f_j(s)}\\ 
&\leq e^{\eps/2} \sum_{s' \text{ s.t. } \Med(h'+ s')=j}\pr{r=s'} \\
&=  e^{\eps/2}\cdot\pr{\M(h')= \ell_j}.
\end{align*}
A symmetric argument shows this is also true switching $h$ and $h'$, which completes the proof.
\end{proof}

We note that the only property the proof uses about the noise distribution is that $\Pr[r=s] = e^{\eps/2}\cdot \Pr[r=f_j(s)]$.  This property
does not hold for Xiao's noise distribution as described, due to it being truncated above at $2t$, but would hold if his noise distribution was truncated only below.

We next show that this mechanism is truthful
and individually rational.

\begin{theorem} 
\label{thm:facility}
Mechanism~\ref{mech:dfac} is universally truthful and individually rational
for player $i$ provided that, for some function $F_i$:
\begin{squishenum}
\item Player $i$'s privacy utility $\Up[i]$ satisfies Assumption~\ref{ass:generic-privacy} with privacy bound function $F_i$, and
\item For all $o,o'\in \Theta$ such that $\theta_i<o<o'$ or $o'>o>\theta_i$, we have
$\Uo[i](\theta_i,o)-\Uo[i](\theta_i,o') \geq 2 F_i(e^\eps).$
\end{squishenum}
In particular, if all players share the standard outcome utility function $\Uo[i](\theta_i,o)=-|\theta_i-o|$ and have the same privacy bound
function $F_i=F$, then the mechanism is universally truthful and individually rational provided that
$$\min_{j\neq k} |\ell_j-\ell_k| \geq 2F(e^\eps).$$
\end{theorem}
So, for a fixed set $\Theta$ of player types (preferred locations), we can take $\eps$ to be a small constant and have truthfulness and individual rationality.


\begin{proof}

Fix $r\in\N^q$, the randomness used by the mechanism and the reports $\theta_{-i}$ of other players.  Following
Xiao~\cite{Xiao11}, we think of $r$ as representing the reports of some fictional additional players, and follow the
truthfulness reasoning for the standard, noiseless median mechanism.
Suppose $\M(\theta_i,\theta_{-i};r) = o$ and  $\M(\theta_i',\theta_{-i};r) = o' \neq o$.
If $\theta_i < o$, then no other report of player $i$ can reduce the median,
so we must have $o'>o$.
Thus, this change has moved the facility at least one location away from $i$'s preferred location.
Similarly, if $\theta_i > o$, we have $o' < o$ so again the change is away from $i$'s preferred location.
Therefore, universal truthfulness follows by Lemma~\ref{lem:universal}.
For individual rationality, we can model non-participation as a report of a type $\bot$ that does not get included in the histogram.
Again, any change of of the median caused by reporting $\bot$ will move it away from $i$'s preferred location.  Thus $\M$ is individually rational.
\end{proof}


\begin{proposition} \label{prop:facility-efficiency}
Suppose that every player $i$ has the standard outcome utility function $\Uo[i](\theta_i,o)=-|\theta_i-o|$.
Then for every profile of types $\theta\in \Theta^n$, if we choose $o\leftarrow \M(\theta)$ using Mechanism~\ref{mech:dfac}, we have
\begin{squishenum}
\item $\Pr\left[\sum_i \Uo[i](\theta_i,o) \leq \max_{o'}\left(\sum_i \Uo[i](\theta_i,o')\right)-\Delta\right] \leq q\cdot e^{-\eps \Delta/q}.$
\item $\Exp\left[\sum_i \Uo[i](\theta_i,o)\right] \geq \max_{o'}\left(\sum_i \Uo[i](\theta_i,o')\right)- O(q/\eps).$
\end{squishenum}
\end{proposition}

Thus, the social welfare is within $\Delta=\tilde{O}(q)/\eps$ of optimal, both in expectation and with high probability.
Like with Proposition~\ref{prop:voting-efficiency}, these bounds are independent of the number $n$ of participants, so
we obtain asymptotically optimal social welfare as $n\rightarrow \infty$.  Also like the discussion after Proposition~\ref{prop:voting-efficiency}, by taking
$\eps=\eps(n)$ to be such that $\eps=o(1)$ and $\eps=\omega(1/n)$ (e.g., $\eps=1/\sqrt{n}$), the sum of privacy utilities is a vanishing fraction of $n$
(for participants satisfying Assumption~\ref{ass:generic-privacy} with a common privacy bound function $F$).

\begin{proof}
Note that $-\sum_i \Uo[i](\theta_i,o') = \sum_j h_j\cdot |\ell_j-o'|$, where $h=(h_1,\ldots,h_q)$ is the histogram corresponding to
$\theta$.  This social welfare is minimized by taking $o'=\Med(h)$.
Our mechanism, however, computes the optimal location for the noisy histogram $h+r$.
We can relate the two as follows:
\begin{eqnarray*}
-\sum_i \Uo[i](\theta_i,o) &=& \sum_j h_j\cdot |\ell_j-o|\\
&\leq& \sum_j (h_j+r_j)\cdot |\ell_j-o|\\
&=& \min_{o'} \sum_j (h_j+r_j)\cdot |\ell_j-o'|\\
&\leq& \min_{o'} \sum_j h_j\cdot |\ell_j-o'|+\sum_j r_j\\
&=& - \max_{o'} \sum_i \Uo[i](\theta_i,o') + \sum_jr_j.
\end{eqnarray*}
Thus, for the high probability bound, it suffices to bound the probability that $\sum_j r_j\geq \Delta$. 
This in turn is bounded by $q$ times\svc{I suspect we can save the ``union-bound'' factor of $q$ by directly
analyzing the tails of the sum of $q$ exponential random variables} the probability that
any particular $r_j$ is at least $\Delta/q$, which is at most $e^{-\eps\Delta/q}$.
For the expectation bound, we have
$$\Exp[\sum_j r_j] = \sum_j \Exp[r_j] = q\cdot \frac{1}{1-e^{-\eps/2}} = O\left(\frac{q}{\eps}\right).$$
\end{proof}

\ici{In the future, think about the non-discrete version.}
\Version{$ $Id$ $}
\section{General Social Choice Problems With Payments}\label{sec:vcg}

In the preceding two sections we have considered social choice problems where a group needs to choose among a (typically small) set of options
with mechanisms that do not use money.
In this section, we apply our framework social choice problems
where payments are possible using an
adaptation of the Vickrey-Clarke-Groves (VCG)
mechanism. (This is the setting for which the Groves mechanism was originally designed and unlike in auction settings the number of outcomes is independent of the number of players.)
In the general case we examine now, we don't assume any structure on the utility functions
(other than discreteness),
and thus
need to use payments to incentivize players to truthfully reveal their preferences.

Specifically, the type $\theta_i\in \Theta$ of a player will specify
a utility $\Uo(\theta_i,o)\in \{0,1,\ldots,M\}$ for
each outcome $o$ from a finite set $O$.
This could correspond, for example, to players having values for outcomes
expressible in whole dollars with some upper and lower bounds.
This assumption ensures a finite set of types $\Theta$ and that if a player changes his reported value it must change by some minimum amount (1 with our particular assumption).
Note that this formulation still allows players to be indifferent among outcomes.
Gicen our notion of a type,
all players share the same outcome utility function $\Uo[i]=\Uo$
In order to reason about individual rationality, we also assume that the
set of types includes a type $\bot$ that corresponds to not participating
(i.e., $\Uo[i](\bot,o) = 0$ for all $o$ and $i$).
For notational convenience, we assume that $O = \{0,1,\ldots,|O|-1\}$.


Our goal is to
choose the outcome $o^*$ that maximizes social welfare (ignoring privacy), i.e.,
$o^*=\argmax_{o\in O} \sum_i \Uo(\theta_i,o)$.  A standard way to do so is the Groves
mechanism, a special case of the more general VCG mechanism.  Each
player reports his type and then the optimal outcome $o^*$ is chosen based
on the reported types.  To ensure truthfulness, each player is charged
the externality he imposes on others.  If
\svc{dropped ``$o_i$ is the outcome chosen
with his input'' and am using $o^*$ instead for consistency with mechanism below}
$o_{-i}=\argmax_o \sum_{j\neq i} \Uo(\theta_j,o)$ is the outcome that would have
been chosen without $i$'s input, then player $i$ makes a payment of
\begin{equation}
\label{eqn:VCG}
P_i = \sum_{j \neq i} \left(\Uo(\theta_j,o_{-i}) - \Uo(\theta_j,o^*)\right),
\end{equation}
for a combined utility of
$\Uo(\theta_i,o^*)-P_i.$

In addition to subtracting payments from player $i$'s utility as above, we also need
to consider the effect of payments on privacy.  (The modelling in Section~\ref{sec:modelprivacy} did not
consider payments.)
 While it may be reasonable to treat the
payments players make as secret, so that making the payment does not
reveal information to others, the amount a player is asked to
pay reveals information about the reports of {\em other} players.
Therefore, we will require that the mechanism releases some {\em public}
payment information $\pi$ that enables all players to compute
their payments,
i.e., 
the payment $P_i$ of player $i$ should be a function of $\theta_i$, $\pi$, and $o^*$.  For example, $\pi$ could just
be the $n$-tuple $(P_1,\ldots,P_n)$, which corresponds to making all payments public.  But
in the
VCG
mechanism
it suffices for $\pi$ to
include
the value $V_o=\sum_i \Uo(\theta_i,o)$ for all
outcomes $o\in O$, since
\begin{eqnarray*}
P_i = (V_{o_{-i}}-\Uo(\theta_i,o_{-i}))-
(V_{o^*}-\Uo(\theta_i,o^*))
= max_o \left((\Uo(\theta_i,o^*)-\Uo(\theta_i,o))-(V_{o^*}-V_{o})\right),
\end{eqnarray*}
which can be computed using just the $V_o$'s, $o^*$, and $\theta_i$.  Moreover, we actually only need to release
the differences $V_{o^*}-V_o$, and only need to do so for outcomes $o$ such that $V_{o^*}-V_o\leq M$, since only such outcomes have a chance of
achieving the above maximum. (Recall that $\Uo(\theta_i,o)\in \{0,1,\ldots,M\}$.)
This observation forms the basis of our mechanism, which we will show to be
truthful for players that value privacy (under Assumption~\ref{ass:generic-privacy}).

Before stating our mechanism, we summarize how we take payments into account in our modelling.
Given reports $\theta'\in \Theta^n$ and randomness $r$, our mechanism $\M(\theta';r)$ outputs a pair
$(o^*,\pi)$, where $o^*\in O$ is the selected outcome and $\pi$ is ``payment information''.  Each player then
should send payment $P_i=P(\theta_i',o^*,\pi)$  to the mechanism.  (The payment function $P$ is something
we design together with the mechanism $\M$.)
If player $i$'s true type is $\theta_i$, then her total utility is:\svc{replaced $P_i$ with $P(\cdots)$, and consequently
added $\pi$ and $\theta'_i$ to
arguments of $U_i$}
$$U_i(\theta_i,o^*,\pi,\M,\theta')=\Uo(\theta_i,o^*)-P(\theta_i',o^*,\pi)+\Up[i](\theta_i,(o^*,\pi),\M,\theta'_{-i}).$$
Note that we measure the privacy of the {\em pair} $(o^*,\pi)$, since both are released publicly.

To achieve truthfulness for players that value privacy, we will modify the VCG mechanism described above by
adding noise to the values $V_o$.  This yields the following mechanism:

\begin{mechanism}\label{mech:vcg}
Differentially private VCG mechanism

\noindent Input: profile $\theta\in \Theta^n$ of types, privacy parameter $\eps>0$.
\begin{squishenum}
\item Choose $\lambda_o$ from a (discrete) Laplace distribution for each outcome $o$.  Specifically,
we set $\Pr[\lambda_o = k] \propto \exp(- (\eps\cdot |k|)/(M\cdot |O|))$ for every integer $k\in \Z$.\svc{for
differential privacy, we only need to divide by $M$ in the exponent - the maximum amount by which a single
player can change a $V_o$}
\item Calculate values $V_o =  \sum_j \Uo(\theta_j,o) + \lambda_o + o/|O|$ for each outcome $o$.  (Recall that we
set $O=\{0,\ldots,|O|-1\}$. The $o/|O|$ term is introduced in order to break ties.)\svc{added parenthetical comment}
\item Select outcome $o^* = \arg\max_o V_o$.
\item Set the payment information $\pi = \{ (o,V_{o^*}-V_o) :  V_o\geq V_{o^*}-M\}$.
\item Output $(o^*,\pi)$.
\end{squishenum}
Each player $i$ then sends a payment of
$P_i = P(\theta_i,o^*,\pi) =   \max_o \left((\Uo(\theta_i,o^*)-\Uo(\theta_i,o))-(V_{o^*}-V_{o})\right).$
\end{mechanism}

By standard results on differential privacy, the tuple of noisy values $\{V_o\}$ is $\eps$-differentially private.  Since the output
$(o^*,\pi)$ is a function of the $V_o$'s, the output is also differentially private:

\begin{lemma}
\label{lem:VCGpriv1}
Mechanism~\ref{mech:vcg} is $\eps$-differentially private.
\end{lemma}

We now prove that the mechanism is truthful in expectation for players that value privacy (satisfying Assumption~\ref{ass:generic-privacy}).
To do this, we use Lemma~\ref{lem:SD}, which shows
that by taking $\eps$ sufficiently small, the expected change in privacy utility from misreporting $\theta_i'$ instead of $\theta_i$ can be
made an arbitrarily small fraction of the statistical difference
$\SD(\M(\theta_i,\theta_{-i}),\M(\theta_i',\theta_{-i}))$.  Thus, to show truthfulness in expectation, it suffices to show that
the statistical difference is at most a constant factor larger than the
expected decrease in utility from misreporting.  That is, we want to show:
\begin{eqnarray*}
\lefteqn{\SD(\M(\theta_i,\theta_{-i}),\M(\theta_i',\theta_{-i}))}&&\\
&=&
O\left(\Exp[\Uo(\theta_i, \M(\theta_i,\theta_{-i}))-P(\theta_i,\M(\theta_i,\theta_{-i}))] -
\Exp[\Uo(\theta_i, \M(\theta_i',\theta_{-i}))-P(\theta_i',\M(\theta_i',\theta_{-i}))]\right).
\end{eqnarray*}

To bound the statistical difference, we write $\M(\theta;r)=(\M^1(\theta;r),\M^2(\theta;r))$, where $\M^1$ gives the outcome $o^*$ and $\M^2$
gives the payment information $\pi$.  Then we have:
\begin{align*}
\SD(\M(\theta_i,\theta_{-i}),\M(\theta_i',\theta_{-i}))
&\leq \Pr_r[\M(\theta_i,\theta_{-i};r)\neq \M(\theta_i',\theta_{-i};r)]
\leq \Pr_r[\M^1(\theta_i,\theta_{-i};r)\neq \M^1(\theta_i',\theta_{-i};r)]\\
&+ \Pr_r[\M^1(\theta_i,\theta_{-i};r) = \M^1(\theta_i',\theta_{-i};r) \wedge \M^2(\theta_i,\theta_{-i};r) \neq \M^2(\theta_i',\theta_{-i};r)].
\end{align*}

The next lemma bounds the statistical difference coming from the outcome:

\begin{lemma}
\label{lem:VCGout}
\begin{eqnarray*}
\Pr_r[\M^1(\theta_i,\theta_{-i};r)\neq \M^1(\theta_i',\theta_{-i};r)]
&\leq |O| \cdot (
\Exp[\Uo(\theta_i, \M^1(\theta_i,\theta_{-i}))-P(\theta_i,\M(\theta_i,\theta_{-i}))]\\
& -
\Exp[\Uo(\theta_i, \M^1(\theta_i',\theta_{-i}))-P(\theta_i',\M(\theta_i,\theta_{-i}))]).
\end{eqnarray*}
\end{lemma}

\svci{added more detail to beginning and end of proof to help people less familiar with VCG}
\begin{proof}
It suffices to show that for every value of $r$, we have:
\begin{align}
\label{eqn:vcg-universal}
\hspace{-4mm}
I[\M^1(\theta_i,\theta_{-i};r)\neq \M^1(\theta_i',\theta_{-i};r)]
&\leq |O| \cdot (
\Uo(\theta_i, \M^1(\theta_i,\theta_{-i};r))-P(\theta_i,\M(\theta_i,\theta_{-i};r))\\
&-
\Uo(\theta_i, \M^1(\theta_i',\theta_{-i};r))-P(\theta_i',\M(\theta_i,\theta_{-i};r))), \nonumber
\end{align}
where $I[X]$ denotes the indicator for the event $X$.  (Then taking expectation over $r$ yields the desired result.)

If $\M^1(\theta_i,\theta_{-i};r)= \M^1(\theta_i',\theta_{-i};r)$, then both the left-hand and right-hand sides are zero.
(Recall that the payment made by player $i$ on an outcome $o$ depends only on the reports of the other players and the randomness of the mechanism.)

So consider a value of $r$
such that $\M^1(\theta_i,\theta_{-i};r)\neq \M^1(\theta_i',\theta_{-i};r)$ (i.e., where the indicator is 1).
We can treat the $\lambda_o + o/|O|$ term added to each $V_o$ as the report of another player to the standard VCG mechanism.
We know that
\[
\Uo(\theta_i, \M^1(\theta_i,\theta_{-i};r))-P(\theta_i,\M(\theta_i,\theta_{-i};r)) -
\Uo(\theta_i, \M(\theta_i',\theta_{-i};r))-P(\theta_i',\M(\theta_i,\theta_{-i};r))
\geq 0
\]
because VCG is incentive compatible for players who don't have a privacy utility.
Since the mechanism adds an $o/|O|$ term to $V_o$ to avoid ties, the above inequality is strict. Moreover, the left-hand side is at least $1/|O|$, which establishes Inequality (\ref{eqn:vcg-universal}).

In more detail,
let $o^*=\M^1(\theta_i,\theta_{-i};r)$ and $o'=\M^1(\theta_i',\theta_{-i};r)$ for some $o'\neq o^*$.
Write $W_o =  \sum_{j\neq i} \Uo[j](\theta_j,o) + \lambda_o + o/|O|$ for each outcome $o$  ($W_o$ is just $V_o$ excluding the report of player $i$),
and $o_{-i} = \argmax_o W_o$.  Since the mechanism chose $o^*$ on report $\theta_i$, we must have
$$W_{o^*}+\Uo(\theta_i,o^*) \geq W_{o'}+\Uo(\theta_i,o').$$  Since the fractional parts of the two sides are different multiples of $1/|O|$ (namely
$o^*/|O|$ ad $o'/|O|$), we have:
$$W_{o^*}+\Uo(\theta_i,o^*) \geq W_{o'}+\Uo(\theta_i,o')+1/|O|.$$
Thus:
\begin{align*}
&\Uo(\theta_i, \M^1(\theta_i,\theta_{-i};r))-P(\theta_i,\M(\theta_i,\theta_{-i};r))\\
&= \Uo(\theta_i,o^*)-(W_{o_{-i}}-W_{o^*})\\
&\geq \Uo(\theta_i,o')-(W_{o_{-i}}-W_{o'})+1/|O|\\
&= \Uo(\theta_i, \M(\theta_i',\theta_{-i};r))-P(\theta_i',\M(\theta_i,\theta_{-i};r))+1/|O|,
\end{align*}
establishing Inequality~(\ref{eqn:vcg-universal}).
\end{proof}

Now we need to prove a similar bound for the probability of misreporting only affecting the payment information $\pi$.
We note that one trivial solution for handling payments is to only collect payments with a very small probability $p$, but
increase the magnitude of the payments by a factor of $1/p$. In order for payments to not
contribute more to the statistical difference than the outcome, we can take
$p$ to be the minimum possible nonzero value of the probability that a misreport can change the outcome
(i.e., $\Pr_r[\M^1(\theta_i,\theta_{-i};r)\neq \M^1(\theta_i',\theta_{-i};r)]$).
However, this quantity is exponentially small in $n$.
This would make the magnitude of payments exponentially large, which is
undesirable. (Our assumption that players are risk neutral seems
unreasonable in such a setting.)  However, it turns out that we do not
actually need to do this; our mechanism already releases payment information
with sufficiently low probability.
Indeed, we only release payment information relating to an outcome $o$ when $V_o$
is within $M$ of $V_{o^*}$, and the probability that this occurs cannot be much larger than the probability
that the outcome is changed from $o^*$ to $o$.\svc{changed last two sentences}

\begin{lemma}
\label{lem:VCGpriv3}
\begin{eqnarray*}
\lefteqn{\Pr_r[\M^1(\theta_i,\theta_{-i};r) = \M^1(\theta_i',\theta_{-i};r) \wedge \M^2(\theta_i,\theta_{-i};r) \neq \M^2(\theta_i',\theta_{-i};r)]} &&\\
&\leq& 2M e^{\eps/|O|}\cdot \Pr_r[\M^1(\theta_i,\theta_{-i};r)\neq \M^1(\theta_i',\theta_{-i};r)].
\end{eqnarray*}
\svc{replaced $M/(M+1)$ with 1 in exponent}
\end{lemma}

\begin{proof}
First observe that
\begin{align*}
&\Pr_r[\M^1(\theta_i,\theta_{-i};r) = \M^1(\theta_i',\theta_{-i};r) \wedge \M^2(\theta_i,\theta_{-i};r) \neq \M^2(\theta_i',\theta_{-i};r)]\\
&\leq \sum_{o_1 \neq o_2} \Pr_r[\M^1(\theta_i,\theta_{-i};r) = \M^1(\theta_i',\theta_{-i};r)=o_1 \wedge \M^2(\theta_i,\theta_{-i};r) \neq \M^2(\theta_i',\theta_{-i};r) \text{ on } o_2],
\end{align*}
by which we mean that either $(o_2,V_{o_1} -V_{o_2})$ is released in one case but not the other or it is released in both cases but with different values. 

Fix $o_1$ and $o_2$ as above.
If $\Uo(\theta_i,o_1) - \Uo(\theta_i,o_2) = \Uo(\theta_i',o_1) - \Uo(\theta_i',o_2)$, then
$\Pr_r[\M^1(\theta_i,\theta_{-i};r) = \M^1(\theta_i',\theta_{-i};r)=o_1 \wedge \M^2(\theta_i,\theta_{-i};r) \neq \M^2(\theta_i',\theta_{-i};r) \text{ on } o_2] = 0$ because the difference between $V_{o_1}$ and $V_{o_2}$ is not changed by the misreporting.
So assume that $\Uo(\theta_i,o_1) - \Uo(\theta_i,o_2) \neq \Uo(\theta_i',o_1) - \Uo(\theta_i',o_2)$; these values must differ
by at least 1 due to the discreteness assumption.
Fix $\lambda_o= k_o$ for $o \neq o_2$. Denote them as a vector $\lambda_{-o_2}=k_{-o_2}$. Consider some value $k_{o_2}$ such that when
$\lambda_{o_2} =k_{o_2}$ we have
$\M^1(\theta_i,\theta_{-i};(k_{o_2},k_{-o_2})) = \M^1(\theta_i',\theta_{-i};(k_{o_2},k_{-o_2}))=o_1$ and $\M^2(\theta_i,\theta_{-i};(k_{o_2},k_{-o_2})) \neq \M^2(\theta_i',\theta_{-i};(k_{o_2},k_{-o_2})) \text{ on } o_2$. (If there is no such $k_{o_2}$ then the event has probability 0 for this choice of $k_{-o_2}$.) Now consider increasing the value of $\lambda_{o_2}$.
Let $\hat{k}_{o_2}$ be the minimum value such that either
$\M^1(\theta_i,\theta_{-i};(\hat{k}_{o_2},k_{-o_2})) = o_2$
or
 $\M^1(\theta_i',\theta_{-i};(\hat{k}_{o_2},k_{-o_2})) = o_2$.
At the first such value of $\hat{k}_{o_2}$, only one of these two events will happen because
$\Uo(\theta_i,o_1) - \Uo(\theta_i,o_2)$ and $\Uo(\theta_i',o_1) - \Uo(\theta_i',o_2)$ differ by at least 1.\svc{minor edits to previous two sentences}
Moreover, we have $\hat{k}_{o_2} \leq k_{o_2} + M$ because with $\lambda_{o_2} = k_{o_2}$ we have $V_{o_1} - V_{o_2} \leq M$ for either report $\theta_i$ or $\theta'_i$.
Since $\Pr[\lambda_{o_2}=k]\propto \exp(-\eps\cdot |k|/(M\cdot |O|))$, we have
$\Pr[\lambda_{o_2} = k_{o_2}] \leq \exp(\eps/|O|)\cdot \Pr[\lambda_{o_2} = \hat{k}_{o_2}]$.  
Furthermore, there can be at most $M$ such values of $k_{o_2}$.  Thus,
\begin{align*}
&\Pr_r[\lambda_{-o_2}=k_{-o_2} \wedge \M^1(\theta_i,\theta_{-i};r) = \M^1(\theta_i',\theta_{-i};r)=o_1 \wedge \M^2(\theta_i,\theta_{-i};r) \neq \M^2(\theta_i',\theta_{-i};r) \text{ on } o_2] \\
&\leq M e^{\eps/|O|}\Pr_r[\lambda_{-o_2}=k_{-o_2} \wedge \M^1(\theta_i,\theta_{-i};r) \neq \M^1(\theta_i',\theta_{-i};r) \wedge \M^1(\theta_i,\theta_{-i};r) \in \{o_1, o_2\}\\
&\wedge \M^1(\theta'_i,\theta_{-i};r) \in \{o_1, o_2\}]
\end{align*}
Summing over all $o_1 \neq o_2$ and $k_{-o_2}$ gives us the lemma. The factor 2 in the lemma statement is due to the fact that
\begin{align*}
&\hspace{-5mm}\sum_{o_1 \neq o_2, k_{o_2}}\Pr_r[\lambda_{-o_2}=k_{-o_2} \wedge \M^1(\theta_i,\theta_{-i};r) \neq \M^1(\theta_i',\theta_{-i};r) \wedge \M^1(\theta_i,\theta_{-i};r) \in \{o_1, o_2\} \wedge \M^1(\theta'_i,\theta_{-i};r) \in \{o_1, o_2\}]\\
& = 2 \Pr_r [\M^1(\theta_i,\theta_{-i};r) \neq \M^1(\theta_i',\theta_{-i};r)].
\end{align*}
\end{proof}

Combining Lemmas \ref{lem:VCGout} and \ref{lem:VCGpriv3}, we have
\begin{align*}
\SD(\M(\theta_i, \theta_{-i}), \M(\theta'_i, \theta_{-i}))
&\leq |O|\cdot (1+ 2M e^{\eps/|O|}) \cdot (
\Exp[\Uo(\theta_i, \M^1(\theta_i,\theta_{-i}))-P(\theta_i,\M(\theta_i,\theta_{-i}))]\\
&-
\Exp[\Uo(\theta_i, \M^1(\theta_i',\theta_{-i}))-P(\theta_i',\M(\theta_i,\theta_{-i}))]).
\end{align*}
Applying Lemma~\ref{lem:SD} gives us our theorem.

\begin{theorem} \label{thm:vcg}
Mechanism~\ref{mech:vcg} is truthful in expectation and individually rational
for player $i$ provided that, for some function $F_i$:
\begin{squishenum}
\item Player $i$'s privacy utility $\Up[i]$ satisfies Assumption~\ref{ass:generic-privacy} with privacy bound function $F_i$, and
\item $2F_i(e^{\eps})\cdot |O|\cdot (1+ 2M e^{\eps/|O|}) \leq 1$.
\end{squishenum}
In particular, if all players have the same privacy bound function $F_i=F$, it suffices to take $\eps$ to be a sufficiently small
constant depending only on $M$ and $|O|$ (and not the number $n$ of players).\svc{new sentence}
\end{theorem}

Truthfulness in expectation  relies on players being risk
neutral in terms of their privacy utility so that it is acceptable that with some low probability,
the privacy costs are larger than their utility from the outcome.
An alternative approach that does not rely on risk neutrality
is to switch from the VCG mechanism to the Expected
Externality mechanism.  This is a variant on VCG that, rather than
charging players the actual externality they impose as
in Equation~\eqref{eqn:VCG}, charges them their expected externality
\begin{equation}
\label{eqn:EE}
E_{\theta \sim p}\left[\sum_{j \neq i} \Uo(\theta_j,o_{-i}) - \Uo(\theta_j,o^*)\right],
\end{equation}
where $p$ is a prior distribution over $\Theta^n$, $o_{-i}$ is the outcome that maximizes the sum of outcome utilities of players other than $i$, and $o^*$ is the outcome that maximizes the sum of outcome utilities when $i$ is included.
Essentially, $i$ is charged the expected amount he would have to pay under VCG given the prior over types.
Since the amount
players are charged is independent of the actual reports of others,
collecting payments has no privacy implications.  (The proof of Lemma~\ref{lem:VCGout} shows that
if we only consider the privacy cost of the outcome, then we have universal truthfulness.)
However, the use of
a prior means that the truthfulness guarantee only holds in a
Bayes-Nash equilibrium.  On the other hand, this mechanism does have
other nice properties such as being adaptable to guarantee budget
balance.

Finally, we show that Mechanism~\ref{mech:vcg} approximately preserves VCG's efficiency.

\begin{proposition} \label{prop:vcg-efficiency}
For every profile of types $\theta\in \Theta^n$, if we choose $o\leftarrow \M(\theta)$ using Mechanism~\ref{mech:vcg}, then we have:
\begin{squishenum}
\item
$\Pr\left[\sum_i \Uo[i](\theta_i,o) < \max_{o'}\left(\sum_i \Uo[i](\theta_i,o')\right)-\Delta\right] \leq 2|O| \cdot e^{-\eps \Delta/(2M\cdot |O|)},$
\item
$\Exp\left[\sum_i \Uo[i](\theta_i,o)\right] \geq \max_{o'}\left(\sum_i \Uo[i](\theta_i,o')\right) - O(|O|^2\cdot M/\eps).$
\end{squishenum}
\end{proposition}
\svci{corrected and expanded proposition and added more details to proof}

\begin{proof}
Let $o^{**}= \argmax_o \Uo[j](\theta_j,o)$. For the output $o^*$ of Mechanism~\ref{mech:vcg}, we have:
\begin{eqnarray*}
\sum_j \Uo[j](\theta_j,o^*)
&=& V_{o^*} - \lambda_{o^*}-o^*/|O|\\
&\geq& V_{o^{**}} - \lambda_{o^*}-o^*/|O|\\
&=& \left(\max_o \Uo[j](\theta_j,o)\right)+\lambda_{o^{**}}+o^{**}/|O|-\lambda_{o^*}-o^*/|O|\\
&>& \left(\max_o \Uo[j](\theta_j,o)\right) - \max_{o} (\lambda_o-\lambda_{o^{**}}) - 1.
\end{eqnarray*}
So we are left with bounding $\max_{o} (\lambda_o-\lambda_{o^{**}})$ for random variables $\lambda_o$ such that
$\Pr[\lambda_o=k]\propto \exp(-\eps\cdot |k|/(M\cdot |O|))$.
For each $o$,
\begin{eqnarray*}
\Pr[\lambda_o-\lambda_{o^{**}}\geq \Delta]
\leq \Pr[\lambda_o \geq \Delta/2]+\Pr[\lambda_{o^**}\leq -\Delta/2]
\leq 2\exp(-\eps \Delta/(2M\cdot |O|)).
\end{eqnarray*}
Taking a union bound over the choices for $o$ completes the high probability bound.
For the expectation, we have:
\begin{equation*}
\Exp[\max_{o} (\lambda_o-\lambda_{o^{**}})]
\leq \Exp\left[\sum_o|\lambda_o|\right]
= |O|\cdot O\left(M\cdot |O|/\eps\right).
\end{equation*}
\end{proof}

Thus, the social welfare is within $\tilde{O}(|O|^2)\cdot M/\eps$ of optimal, both in
expectation and with high probability.
Like with Proposition~\ref{prop:voting-efficiency}, these bounds are independent of the number $n$ of participants, so
we obtain asymptotically optimal social welfare as $n\rightarrow \infty$.  Also like the discussion after Proposition~\ref{prop:voting-efficiency}, by taking
$\eps=\eps(n)$ to be such that $\eps=o(1)$ and $\eps=\omega(1/n)$ (e.g., $\eps=1/\sqrt{n}$), the sum of privacy utilities is also a vanishing fraction of $n$ (for participants satisfying Assumption~\ref{ass:generic-privacy} with a common privacy bound function $F$).

\Version{$ $Id$ $}
\section{Discussion} \label{sec:discussion}

We conclude by discussing a Bayesian interpretation of our privacy model and several of the model's limitations.




Our modelling of privacy in Section~\ref{sec:modelprivacy} is motivated in part by viewing privacy as a concern
about {\em other's beliefs about you}.  Fix a randomized mechanism $\M : \Theta^n \times \Rand \rightarrow O$, a player $i\in [n]$, and a profile  $\theta_{-i}\in \Theta^{n-1}$ of other player's reports.  Suppose that an adversary has a prior $T_i$ on
the type of player $i$, as well as a prior $S_i$ on the strategy $\sigma : \Theta\rightarrow \Theta$ played by player $i$.
Then upon seeing an outcome $o$ from the mechanism, the adversary should replace $T_i$ with a posterior $T'_i$
computed according to Bayes' Rule as follows: 
\begin{eqnarray*}
\Pr[T'_i=\theta_i] &=&  \Pr[T_i=\theta_i | \M(S_i(T_i),\theta_{-i})=o]\\
&=& \Pr[T_i=\theta_i] \cdot \frac{\Pr[\M(S_i(T_i),\theta_{-i})=o | T_i=\theta_i]}{\Pr[\M(S_i(T_i),\theta_{-i})=o]}.
\end{eqnarray*}
Thus if we set
$x=\max_{\theta',\theta''\in\Theta} (\pr{\M(\theta',\theta_{-i})=o}/\pr{\M(\theta'',\theta_{-i})=o})$
(the argument of $F_i$ in Assumption~\ref{ass:generic-privacy}), then we have
$$x^{-1}\cdot \Pr[T_i=\theta_i] \leq \Pr[T'_i=\theta_i] \leq x\cdot \Pr[T_i=\theta_i].$$
So if $x$ is close to 1, then the posterior $T_i'$ is close to the prior $T_i$, having the same probability mass functions within a factor of $x$, and consequently having statistical difference at most $x-1$.
Thus, Assumption~\ref{ass:generic-privacy} can be justified by asserting that ``if an adversary's beliefs about player $i$
do not change much, then it has a minimal impact on player $i$'s privacy utility.'' One way to think of this is that player $i$ has some smooth value function of the adversary's beliefs about her, and her privacy utility is the difference of the value function after and before the Bayesian updating.
%
This
reasoning follows the lines of Bayesian interpretations of differential privacy due to Dwork and McSherry, and described in \cite{KasiviswanathanSm08}.

This Bayesian modelling also explains why we do not include the strategy played by $i$ in the privacy utility function $\Up[i]$. How a Bayesian adversary updates its beliefs about player $i$ based on the outcome do not depend on the actual strategy played by $i$, but rather on the adversary's beliefs about that strategy, denoted by $S_i$ in the above discussion.  Given that our mechanisms are truthful, it is most natural to consider $S_i$ as the truthful strategy (i.e., the identity function).  If the Bayesian adversary values possessing a correct belief about players, this is analogous to a notion of equilibrium.  If we treat the adversary as another player then if the players report truthfully and the adversary assumes the players report truthfully each is responding optimally to the other. However, if
player $i$ can successfully convince the adversary that she will follow some other strategy $S_i$, then this can be implicitly  taken into
account in $\Up[i]$.  (But if player $i$ further deviates from $S_i$, this should not be taken into account, since
the adversary's beliefs will be updated according to $S_i$.)

Our modelling of privacy in terms of other's beliefs is subject to several (reasonable) critiques:
\squishlist
\item Sometimes a small, continuous change in beliefs can result in discrete choices that have a large impact in someone's life.  For example, consider a ranking of potential employees to hire, students to admit, or suitors to marry---a small change in beliefs about a candidate may cause them to drop one place in a ranking, and thereby not get hired, admitted, or married.
On the other hand, the candidate typically
does not know exactly where such a threshold is and so from their perspective the small change in
beliefs could be viewed as causing a small change in the probability of rejection.
\item Like in differential privacy, we only consider an adversary's beliefs about player $i$ {\em given the rest of the database}.  (This is implicit in us considering a fixed $\theta_{-i}$ in Assumption~\ref{ass:generic-privacy}.)  If
    an adversary believes that player $i$'s type is correlated with the other players (e.g., given by a joint prior $T$ on $\Theta^n$), then conditioning on $T_{-i}=\theta_{-i}$ may already dramatically change the adversary's beliefs about player $i$.  For example, if the adversary knew that all $n$ voters in a given precinct prefer the same candidate (but don't know which candidate that is), then conditioning on $\theta_{-i}$ tells the adversary who player $i$ prefers. We don't measure the (dis)utility for leaking this kind of information.  Indeed, the differentially private election mechanism of Theorem~\ref{thm:voting} will leak the preferred candidate in this example (with high probability).
\item The word ``privacy'' is used in many other ways.  Instead of being concerned about other's beliefs, one may be concerned about self-representation (e.g., the effect that reporting a given type may have on one's self-image).
\squishend

\section*{Acknowledgments}\svc{please add/edit}
This work was inspired by discussions under the Harvard Center Research for Computation and Society's ``Data Marketplace'' project.  We are grateful to the
other participants in those meetings, including Scott Kominers, David Parkes, Felix Fischer, Ariel Procaccia, Aaron Roth, Latanya Sweeney, and Jon Ullman.  We also thank
Moshe Babaioff and Dave Xiao for helpful discussions and comments.


\appendix
\Version{$ $Id$ $}

\DeclareRobustCommand{\secfacility}{\ref{sec:facility}}
\DeclareRobustCommand{\secvcg}{\ref{sec:vcg}}
\DeclareRobustCommand{\xiaofoot}{\footnote{Subsequent to our work, Xiao has revised his model to use a different, prior-free measure of privacy.  This appendix provides a comparison to his original formulation.}}

\section{Comparison to Xiao's Privacy Measure\xiaofoot} \label{sec:xiao}


Xiao~\cite{Xiao11} measures privacy cost as being proportional to the mutual information between a player's type and the outcome of the mechanism,
where the {\em mutual information} between two jointly distributed random variables $X$ and $Y$ is defined to be
$$I(X;Y)=H(X)+H(Y)-H(X,Y)=\Exp_{(x,y)\sim (X,Y)}\left[\log \frac{\Pr[(X,Y)=(x,y)]}{\Pr[X=x]\cdot \Pr[Y=y]}\right],$$
where $H(Z)=\Exp_{z\sim Z}[\log(1/\Pr[Z=z])]$ is Shannon entropy.  In order for the mutual information to make sense, Xiao assumes
a prior $T_i$ on a player's type and the privacy cost also depends on the strategy $\sigma_i : \Theta\rightarrow \Theta$ played by player $i$.  Accordingly
his measure of outcome utility also takes an expectation over the same prior $T_i$, resulting in the following definition.

\begin{definition}
Let $\Theta$ be a type space, $O$ an outcome space, $\Uo : \Theta\times O\rightarrow \R$ an outcome-utility function, and let
$\nu_i\geq 0$ be a measure of player $i$'s value for privacy, and let $T_i$ be a prior on player $i$'s type.  Then a randomized mechanism $\M : \Theta^n\times \Rand\rightarrow O$ is
{\em Xiao-truthful} for player $i$ if for all strategies $\sigma_i : \Theta\rightarrow \Theta$, and all profiles $\theta_{-i}$ of
reports for the other players, we have:
$$\Exp[\Uo(T_i,\M(T_i,\theta_{-i}))]-\nu_i \cdot I(T_i;\M(T_i,\theta_{-i})) \geq
\Exp[\Uo(T_i,\M(\sigma_i(T_i),\theta_{-i}))]-\nu_i \cdot I(T_i;\M(\sigma_i(T_i),\theta_{-i})),$$
where the expectations and mutual information are taken both over $T_i$ and the random choices of $\M$.
\end{definition}

While mutual information is a natural first choice for measuring privacy, it has several disadvantages compared to our modelling:
\begin{itemize}
\item It treats all bits of information the same, whereas clearly one may have different concerns for different aspects of one's private type.  For example,
one may be a lot more sensitive about the high-order bits of one's salary than the low-order bits.
\item It forces us to consider a prior on a player's type and take expected utility over that prior.  Contrast this with the Bayesian interpretation of
our privacy modelling described in Section~\ref{sec:discussion}.  There the prior $T_i$ is only an adversary's beliefs about player $i$'s type, which may be completely incorrect.  Player $i$'s utility is computed with respect to his fixed, actual type $\theta_i$.
\end{itemize}

As mentioned earlier, Xiao's modelling is not a special case of ours, particularly because his modelling of privacy depends on the actual strategy $\sigma_i$ followed by player $i$.  Nevertheless, we can show that truthfulness with respect to our definitions implies truthfulness with respect to his:

\begin{theorem} \label{thm:implyXiao}
If $\M$ is truthful in expectation for player $i$ with respect to the privacy utility function
$$\Up[i](\theta_i,o,\M,\theta_{-i}) = -\nu_i \cdot \log\frac{\Pr[\M(\theta_i,\theta_{-i})=o]}
{\Pr[\M(T_i,\theta_{-i})=o]},$$
then
$\M$ is Xiao-truthful for player $i$ with prior $T_i$.
\end{theorem}

We note that the privacy utility function in Theorem~\ref{thm:implyXiao} satisfies Assumption~\ref{ass:generic-privacy}
with $F_i(x) = \nu_i\cdot \log(x)$, and hence all of our truthful mechanisms are also Xiao-truthful.

\begin{proof}
First note that, by Bayes' Rule,
\begin{equation}
\label{eqn:XiaoUp}
\Up[i](\theta_i,o,\M,\theta_{-i}) = -\nu_i \cdot \log\frac{\Pr[\M(T_i,\theta_{-i})=o | T_i=\theta_i]}
{\Pr[\M(T_i,\theta_{-i})=o]}
= -\nu_i \cdot \log\frac{\Pr[(T_i,\M(T_i,\theta_{-i}))=(\theta_i,o)]}
{\Pr[T_i=\theta_i]\cdot \Pr[\M(T_i,\theta_{-i})=o]}.
\end{equation}
Thus,
\begin{equation}
\label{eqn:Ihonest}
-\nu_i \cdot I(T_i;\M(T_i,\theta_{-i})) =
\Exp\left[\Up[i](T_i,\M(T_i,\theta_{-i}),\M,\theta_{-i})\right].
\end{equation}
To relate the mutual information under strategy $\sigma_i$ to $\Up[i]$, we use
the notion of {\em KL divergence} between two random variables $X$ and $Y$, which is defined as
$$\KL(X||Y) = \Exp_{x\sim X}\left[\log \frac{\Pr[X=x]}{\Pr[Y=y]}\right].$$
We will use the fact that for a random variable $W$ jointly distributed with $X$ and $Y$, we have
$\KL(W,X||W,Y)\geq \KL(X||Y)$.  (This follows from the Log-Sum Inequality~\cite{CoverTh91}.)
Taking $W=T_i$, $X=\M(\sigma_i(T_i),\theta_{-i})$, and $Y=\M(T_i,\theta_{-i})$, we have
\begin{eqnarray*}
\lefteqn{I(T_i;\M(\sigma_i(T_i),\theta_{-i}))} &&\\
&\geq& I(T_i;\M(\sigma_i(T_i),\theta_{-i}))-\KL(T_i,\M(\sigma_i(T_i))||T_i,\M(T_i))+\KL(\M(\sigma_i(T_i))||\M(T_i))\\
&=& \Exp_{(\theta_i,o)\sim (T_i,\M(\sigma_i(T_i),\theta_{-i}))}\left[\log\frac{\Pr[(T_i,\M(T_i,\theta_{-i}))=(\theta_i,o)]}
{\Pr[T_i=\theta_i]\cdot \Pr[\M(T_i,\theta_{-i})=o]}\right].
\end{eqnarray*}
Combining this with Equation~(\ref{eqn:XiaoUp}), we have:
\begin{equation}
\label{eqn:Idishonest}
-\nu_i \cdot I(T_i;\M(\sigma_i(T_i),\theta_{-i})) \leq \Exp\left[\Up[i](T_i,o,\M(\sigma_i(T_i),\theta_{-i})\right].
\end{equation}

By truthfulness in expectation with respect to $\Up[i]$, we have
\begin{eqnarray} \label{eqn:TiE}
\lefteqn{\Exp[\Uo(T_i,\M(T_i,\theta_{-i}))]+\Exp\left[\Up[i](T_i,\M(T_i,\theta_{-i}),\M,\theta_{-i})\right]} && \\
&\geq& \Exp[\Uo(T_i,\M(\sigma_i(T_i),\theta_{-i}))]+\Exp\left[\Up[i](T_i,o,\M(\sigma_i(T_i),\theta_{-i},\theta_{-i})\right] \notag
\end{eqnarray}
Combining Inequalities~(\ref{eqn:Ihonest}), (\ref{eqn:Idishonest}), and (\ref{eqn:TiE}) completes the proof.
\end{proof}

\bibliographystyle{plain}
\bibliography{biblio}

\end{document}